\documentclass[aps,physrev,twocolumn,groupedaddress]{revtex4-2}

\usepackage{graphicx}
\usepackage{amsmath,amssymb,amsthm,mathtools}
\usepackage{physics}
\usepackage{hyperref}
\usepackage{cleveref}
\usepackage{color}

\theoremstyle{plain}
\newtheorem{theorem}{Theorem}
\newtheorem*{thmeorem*}{Theorem}

\theoremstyle{plain}
\newtheorem{lemma}{Lemma}
\newtheorem*{lemma*}{Lemma}

\theoremstyle{plain}

\newtheorem*{proposition*}{Proposition}

\theoremstyle{plain}

\newtheorem*{conjecture*}{Conjecture}

\newtheorem*{corollary*}{Corollary}

\theoremstyle{definition}
\newtheorem{definition}{Definition}

\theoremstyle{plain}

\begin{document}


\title{Error-tolerant secure key leasing for quantum decryption keys in public-key encryption}


\author{Duo XU}
\email[]{xu.duo.x3@s.mail.nagoya-u.ac.jp}
\affiliation{Graduate School of Informatics, Nagoya University,
 Furo-cho, Chikusa-ward, Nagoya-City, 464-8601}
\author{Yuki Takeuchi}
\email[]{Takeuchi.Yuki@bk.MitsubishiElectric.co.jp}
\affiliation{Information Technology R\&D Center, Mitsubishi Electric Corporation, 5-1-1 Ofuna, Kamakura, Kanagawa 247-8501, Japan}


\date{\today}

\begin{abstract}
We propose the first error-tolerant secure key leasing (SKL) for public-key encryption{.} As with SKL in prior works, our protocol consists of a lessor {and} lessee. In the protocol, the lessor encodes its secret key into quantum states and leases the key to the lessee. Then, the lessor can ask the lessee to return the secret key at a later point. The lessor is able to check whether the lessee has returned its key honestly. However, our protocol works even when the leased secret key is subject to noise. The lessee decrypts the ciphertext correctly, and the lessor verifies the return of the secret key correctly when the amount of error is below a certain threshold. Our improved protocol does not change the encoding of the secret key, and thus adds no overhead to the quantum information processing{.}
Our most significant result is a framework to analyze the trade-off between robustness against error and security. We bridge the security of the error-tolerant SKL and that of the error-tolerant certified deletion { with} shortened codes, which is a relatively less explored concept in coding theory.
\end{abstract}


\maketitle

\section{introduction}
Secure key leasing (SKL)~\cite{KN22} is a quantum cryptographic primitive that cannot be realized in any classical manner.
This protocol is conducted between two parties Alice (a Lessor) and Bob (a Lessee). 
Suppose that Alice has a secret key to issue its signature.
When it is on vacation, it would like to delegate this task to Bob.
To this end, it has to lease its secret key to Bob, but after the vacation, it must verify that its key is indeed deleted and is not copied by Bob.
If the key is a string of bits, in principle, the cloning cannot be detected.
The SKL overcomes this problem by encoding the secret key in a quantum state~\footnote{There are also SKL protocols in which Bob sends back the secret key to Alice. In this case, Alice checks whether the returned key is altered. For simplicity, we only consider Bob who destroys the key in this paper.}.
This primitive has a wide range of applications in real life, as it is common to borrow something from others and then certainly return it later.
The SKL employs a quantum mechanical property to detect Bob's malicious actions.

The original SKL protocol was proposed in \cite{KN22}.
Since then, it has been improved in several dicrections~\cite{AKN+23,CGJL23,KMY24,PWY+25,KNP25,KNP25b}.
We here review some of them.
In \cite{CGJL23,PWY+25}, it was shown that a classical Lessor is sufficient to realize PKE-SKL (i.e., SKL for public-key encryption (PKE)).
These results eliminate quantum communication between the Lessor and Lessee.
Ref.~\cite{KMY24} developed a framework to design SKL protocols and extended the applicable range of SKL from PKE to pseudorandom functions and digital signatures.
However, the SKL protocol in \cite{KMY24} is vulnerable if a key is simultaneously leased to multiple Lessees, and then they collaborate to copy the leased key by sharing their keys.
This problem was solved in \cite{KNP25,KNP25b}.

In this paper, we give another improvement.
As a common matter, all above protocols are prone to errors in the Lessor's quantum computer and the communication channel, which hinders the SKL from being practical.
Although the errors in the communication channel may be mitigated by replacing a quantum channel with a classical one as in~\cite{CGJL23,PWY+25}, the errors in the Lessor's quantum computer are still problematic.
These errors can be simultaneously corrected by using quantum error-correcting codes (QECCs), but it would be intriguing whether QECCs are necessary or overkill. 
We show that classical error-correcting codes are sufficient for error-tolerant PKE-SKL.
To this end, we modify the deletion verification process and the ciphertext of \cite{KMY24}.
These modifications change only the classical processing, and thus our idea would improve the practicality of SKL without increasing the demand on quantum resources.

As with the original SKL protocol in \cite{KN22}, our protocol also leverages quantum communication to transmit a key from a Lessor to a Lessee.
Since classical communications are sufficient to implement (non error-tolerant) SKL~\cite{CGJL23,PWY+25,TX25,KLYY25}, someone may think that the necessity of quantum communications merely degrades the practicality of our protocol.
This is not the case.
Existing approaches for replacing quantum communications with classical communications require the Lessor to calculate claw-free functions in superposition.
The precise realization of such coherent calculations would be hard for current experiments~\cite{ZKL+23}.
By leveraging quantum communications, we can avoid this process and ease the burden of the Lessor.

There are three key aspects we should consider regarding error-tolerant PKE-SKL.
The first is that it must {\it certify the deletion when errors are present}.
In other words, if the Lessee follows the honest procedure to delete the leased key, the Lessor must judge that the deletion is executed correctly even when the leased key state is exposed to a slight amount of noise.
However, the original verification process in \cite{KMY24} checks whether the certificate generated by the honest lessee is compatible with an ideal leased key state, and hence it may reject certificates generated from slightly corrupted key states.
To overcome this problem, the error-tolerant PKE-SKL must set a threshold on the difference between actual noisy and ideal certificates below which the verification process accepts any noisy certificate.
We concretely calculate this threshold for the concatenated Hamming codes without compromising the security.

The second is that the Lessee must be able to decrypt the ciphertext even when a key state is slightly corrupted by noises.
In \cite{KMY24}, a uniformly random message $m$ is sent for decryption, which is vulnerable to noises because it is not encoded.
We simply encode $m$ using classical error-correcting codes.

The third is the security.
It is necessary for the error-tolerant PKE-SKL to satisfy the first and second aspects without compromising the security.
This simultaneously achievement would be nontrivial in the sense that these three aspects are in trade-off as follows:
\begin{itemize}
    \item The SKL protocol in \cite{KMY24} is secure but not tolerable to errors.
    \item Assume a protocol in which the message is one bit $b \in \{0,1\}$, and it is encoded with the parallel repetition code. This protocol is trivially error tolerant but is absolutely insecure. As explained below, an adversary can decrypt the original message without fail, even after the adversary returns the secret key.
\end{itemize}

Let the key state be $\otimes_{i=1}^n\ket{sk_i}$.
The adversary chooses $i \in [n]$ uniformly at random.
When the Lessor asks the adversary (i.e., a malicious Lessee) to delete the key, the adversary deletes $\ket{sk_j}$ as the honest Lessee does for all $j \neq i$ and generates their correct certificates.
For the $i$th key state $\ket{sk_i}$, the adversary keeps it and randomly generates its certificate.
If the Lessor allows one or more incorrect certificates among the $n$ ones, the adversary can pass the verification process with unit probability.
Even after this verification process, the adversary can decrypt any new ciphertext $ct = (ct_1, ct_2, \dots, ct_n)$.
By using $\ket{sk_i}$, the adversary recovers the $i$th bit of the encoded message from $ct_i$.
Since the message $b$ is encoded with the parallel repetition code, the $i$th bit is $b$ itself.

\subsection{Technical overview}
The key state $\otimes_{i=1}^n\ket{sk_i}$ is encoded randomly in either computational basis or Hadamard basis. When $\ket{sk_i} = \frac{1}{\sqrt 2}(\ket{0,dk_{i,0}} \pm \ket{1,dk_{i,1}})$ where $dk_{i,0}$ and $dk_{i,1}$ are two decryption keys for the public-key encryption, we say $\ket{sk_i}$ is encoded in the Hadamard basis.  When $\ket{sk_i} = \ket{0,dk_{i,0}}$ or $\ket{sk_i} = \ket{1,dk_{i,1}}$, we say $\ket{sk_i}$ is encoded in the computational basis. When encrypting a message $m$, the sender encrypts each bit $m[i]$ with the encryption key $ek_{i,0}$ and $ek_{i,1}$. We denote the resulting ciphertext as $ct_{i,0}$ and $ct_{i,1}$, respectively. The lessee holding $\ket{sk_i}$ runs the classical decryption algorithm with the ciphertext set to $ct_{i,0}$ and $ct_{i,1}$ coherently, then it can certainly obtain the true value of $m[i]$. This is due to the fact that $ct_{i,0}$ and $ct_{i,1}$ are the ciphertexts for the same plaintext. Informally, we say that the information of the decryption key is stored in the computational basis. The honest lessee deletes each $\ket{sk_i}$ by measuring every qubit in the Hadamard basis. We denote the measurement outcome of the first qubit as $e_i$ and the outcome of the rest as $d_i$. The lessor verifies the deletion of $\ket{sk_i}$ in the Hadamard basis as follows. When the lessee measures the qubits except for the first qubit, a Pauli $Z$ error $\sigma_z^{d_i \cdot (dk_{i,0} \oplus dk_{i,1})}$ is applied to the first qubit. Then, the lessor computes $e_i \oplus d_i \cdot (dk_{i,0} \oplus dk_{i,1})$ which cancels the Pauli error. The lessor checks whether the outcome equals $x[i]$ where $x[i]$ is the phase of $\ket{sk_i}$. Basically, the lessor checks whether the lessee performs the Hadamard measurements, which eliminate the information in the computational basis and render the key useless.

The informal explanation above is proven by reducing the security to the certified deletion~\cite{KMY24}. In the certified deletion, there are two participants: a challenger and an adversary. The challenger sends a BB84 state $\ket{x}_\theta$ to the adversary in a random basis $\theta$. The adversary sends $x^\prime$. The challenger checks whether $x$ and $x^\prime$ match on the bits encoded in the Hadamard basis. When the adversary passes the check, the challenger reveals the basis $\theta$ to the adversary. Then, the adversary sends $x^{\prime\prime}$. The challenger checks whether $x$ coincides with $x^{\prime\prime}$ on the bits encoded in the computational basis. Similar to the idea of SKL, the challenger checks whether the adversary obtains the information in the Hadamard basis, to verify that the information in the computational basis is deleted. It is well known that the adversary cannot produce the information in the computational basis and the Hadamard basis at the same time, except for negligible probability. 

We found the error-tolerant certified deletion in \cite{BI20}. In the error-tolerant certified deletion, the challenger allows $x^\prime$ and $x$ to mismatch on a fraction of $\delta$. In addition, the challenger uses an error-correcting code (ECC) $C$ to encode the information in the computational basis (see \Cref{thm:CDP-BI20} for a formal description). In this work, we attempt to reduce the security of the error-tolerant SKL to the error-tolerant certified deletion.

Similarly, the error-tolerant SKL allows the deletion certificates of $\ket{sk_i}$ to mismatch on a fraction of $\delta$. However, there is a significant structural difference between the error-tolerant certified deletion and our SKL. In our SKL protocol, the classical error-correcting code is applied to the entire message $m$ to ensure decryption correctness even in the presence of noise. In contrast, in the certified deletion game, the error correction is applied only to the part of the string $x$ that is encoded in the computational basis. 

We resolve the gap above using the shortened codes, which is a well-known but relatively less-explored concept in coding theory. The shortened code according to the basis $\theta$ is defined as below.
\begin{definition}[Shortened codes (Informal)]
A shortened code is simply a restricted sub-collection of an original code. We obtain the shortened code by deleting the codewords that have non-zero bits in the position to be encoded in the Hadamard basis.
\end{definition}
 We observe that the shortened code according to $\theta$ is a subcode, which is restricted to the computational basis. By replacing the fixed ECC in the error-tolerant certified deletion with the randomly chosen shortened code, we reduce the security of the error-tolerant SKL to the security of the error-tolerant certified deletion. Then, we can analyze the robustness against noise by inspecting the valid $\delta$ and ECC $C$ such that the error-tolerant certified deletion remains secure. We show an example with the concatenated Hamming code in \Cref{sec:pke-skl-ham}.

\subsection{Related topics}
\paragraph*{\bf Certified Deletion} In a certified deletion (CD) protocol, a sender encrypts the plaintext and sends it to a receiver. Then, the sender asks the receiver to delete the ciphertext. Like the Lessor in SKL, the sender is able to verify whether the receiver deletes the ciphertext. Since its proposal, CD has seen great development~\cite{BI20,BK22,KNT23,BKM+23,BJ24,BVK+24}.

\paragraph*{\bf Unclonable Cryptography} Unclonable cryptography (UC) includes quantum money~\cite{Wie83}, unclonable encryption~\cite{BL20}, copy-protection~\cite{Aar09}, etc. UC differs from SKL in the spirit of that, while the adversary for SKL can refuse to delete its key and make a copy, the adversary for UC cannot copy the key even in the worst circumstance.

\section{Preliminaries}
First, we introduce IND-CPA public-key encryption (PKE). In this work, we need only PKE for a single bit to implement our general PKE-SKL. So, IND-CPA security is the weakest possible security in this case. Though it is weak, it is sufficient to construct our protocol.
\paragraph*{\bf IND-CPA PKE:} The PKE works as follows:
\begin{enumerate}
    \item The receiver generates a public key $pk$ and a secret key $sk$. It publishes the public key.
    \item The sender can encrypt a single bit $b\in\{0,1\}$ using the public key.
    \item The receiver can decrypt the ciphertext with the secret key.
\end{enumerate}
The PKE scheme must satisfy the correctness and the IND-CPA security. The correctness states that, for any $b\in\{0,1\}$, when the sender encrypts $b$ and the receiver decrypts the ciphertext, the receiver decrypts correctly except for negligible probability. The security states that for any QPT adversary, it cannot distinguish between the ciphertext for $0$ and $1$.

\begin{definition}[Linear Codes]
    Let $l_m, l_n, d \in \mathbb N$. A $(l_n, l_m, d)$ linear code consists of an Encoder $Enc: \{0,1\}^{l_m} \rightarrow \{0,1\}^{l_n}$ and a Decoder $Dec: \{0,1\}^{l_n} \rightarrow \{0,1\}^{l_m}$. $Enc$ is a linear map that transforms $l_m$-bits messages into $l_n$-bits codewords. $Dec$ is a linear map that transforms $l_n$-bits codewords back to $l_m$-bits messages. $d$ is the minimum Hamming distance between two different codewords. The linear code corrects up to $\lfloor (d - 1)/2\rfloor$-bits errors.
\end{definition}

Besides the definition, we will introduce the parity check matrix and the syndromes, which are essential in our work.

\paragraph*{\bf Parity Check Matrix} For any $(l_n, l_m, d)$ linear code, there is a matrix $H$ in which each entry is an element of $\mathbb F_2$. Alternatively, the codewords of the linear code are denoted as $w \in \{0,1\}^{l_n}$ such that $Hw = \vec 0$.

\paragraph*{\bf Syndromes} Let $C$ be any $(l_n, l_m, d)$ linear code and $H$ be its parity check matrix. For any $w \in \{0,1\}^{l_n}$, the result $Hw$ is called the syndrome. The syndrome $\vec 0$ indicates that no error occurs. The number of different syndromes is $2^{l_n - l_m}$. We can also denote the number of different syndromes using the rank of the parity check matrix. It is denoted as $2^{{\rm rank}(H)}$.

\begin{definition}[Shortened codes]
    \label{dfn:shortened-codes}
    Let $n,m \in \mathbb N$. $C$ is a linear error correction code which encodes $n$-bits message into $m$ bits. For any $\theta \in \{0,1\}^m$, we define the shortened codes $C_{|\theta}$ as follows:
    \begin{equation}
        C_{|\theta} \coloneqq \{y \in \{0,1\}^m \mid \exists x\in\{0,1\}^n, C(x) = y \land \forall i\in {\cal I}, y[i] = 0 \}
    \end{equation}
    where ${\cal I} \coloneqq \{i\in[m] | \theta[i] = 1\}$. Since we do not care about the relation between the message and the codewords, we denote the code as a subspace of binary strings. In addition, we can always omit the $i$th bit that $i \in {\cal I}$, because such $y[i]$ are fixed to $0$.
\end{definition}

\paragraph*{\bf BB84 states} For strings $x \in \{0, 1\}^n$ and $\theta \in \{0, 1\}^n$, let $\ket{x}_\theta$ be defined as: 
$$ \ket{x}_\theta := H^{\theta[1]} |x[1]\rangle \otimes \cdots \otimes H^{\theta[n]} |x[n]\rangle $$
where $H$ is the Hadamard operator, and $x[i]$ and $\theta[i]$ are the $i$th bits of $x$ and $\theta$, respectively. A state of the form $\ket{x}_\theta$ is called a BB84 state.

\paragraph*{ $\bf x_{|\cal S}$} Let $x \in \{0,1\}^{n}$ and ${\cal S} \subseteq [n]$. $x_{|\cal S}$ is the subsequence obtained by extracting every bit $x_i$ such that $i \in \cal S$.

\paragraph*{\bf The security parameter} We denote the security parameter as $\lambda$ in this work. The larger the security parameter is, the more secure the protocol is. For example, the length of keys is a common security parameter.

A key part of our proof is the Certified Deletion Property. To state the property, we state the game ${\sf CDBB84}_{C,\theta}(\delta,\lambda)$ first.

Let $\lambda \in \mathbb N$ be the security parameter. Let $\delta \in [0,1/2]$. Let $n,m \in \mathbb N$ be polynomials in $\lambda$. Let $\Theta$ be the set of $\theta \in \{0,1\}^{n+m}$ with exactly $n$ 1s. Let $C_\lambda$ be an ECC whose codewords are $m$ bits, and $S$ is the set of possible syndromes. Consider the following game ${\sf CDBB84}_{C,\theta}(\delta,\lambda)$ between an adversary $A = (A_0, A_1)$ and the challenger.
    \begin{enumerate} 
    \item The challenger samples $\theta \in \Theta$ and generates $x$ uniformly at random. Let ${\cal I} = \{i \in [n+m] | \theta[i] = 1\}$ and ${\cal \bar I}$ be its complement. The challenger computes the syndrome $s$ of $x_{|\cal \bar I}$ using $C_\lambda$. The challenger sends $\ket{x}_{\theta}$ to $A_0$.
    \item $A_0$ outputs a classical string $x^* \in \{0, 1\}^{n+m}$ and a state $\text{st}$. 
    \item The challenger aborts the game and the adversary loses, if $hw(x^*_{|\cal I} \oplus x_{|\cal I}) \geq \delta n$.
    \item $\theta$, $x_{|\cal I}$, $s$, and $\text{st}$ are sent to $A_1$. $A_1$ outputs a classical string $x^\prime \in \{0, 1\}^{n+m}$. 
    \item The adversary wins if $x^\prime_{|\cal \bar I}= x_{|\cal \bar I}$. Otherwise, it loses.
    \end{enumerate} 

\begin{figure*}[tbp]
    \centering
    \includegraphics[width=0.8\linewidth]{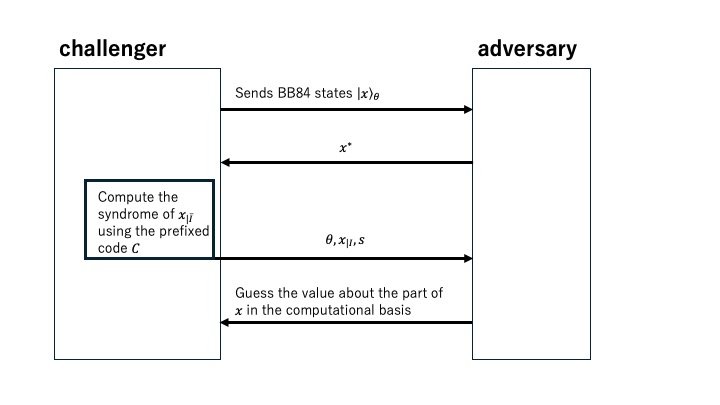}
    \caption{In the figure, we illustrate the certified deletion.}
    \label{fig:CD}
\end{figure*}

\begin{theorem}[Certified Deletion Property~\cite{BI20}]
    \label{thm:CDP-BI20}
    We describe the certified deletion game ${\sf CDBB84}_{C,\theta}(\delta,\lambda)$ in \Cref{fig:CD}. 
    
    For any adversary, it can win the game with probability at most 
    \begin{equation}
        \Pr[\text{The adversary wins}]=\frac{1}{2}\sqrt{2^{-m(1-h(\delta + \nu))+\log |S| + 1}} + 2\epsilon(\nu)
    \end{equation}
    where $\epsilon(\nu)$ is negligible in $m,n$ such that for any $\nu \in (0,1)$
    \begin{equation}
        \epsilon(\nu) = \exp(\frac{-m n^2 \nu^2}{(n+m)(n+1)}).
    \end{equation}
\end{theorem}
We noticed that they did not take into account the affect to the security by the number of syndromes in \cite{BI20}. We fix the flaw above and present a proof in the appendix.

\section{Main results}
\label{sec:proposal}
In this section, we present the construction of the error-tolerant PKE-SKL (see \Cref{sec:cons}). Then, we establish the trade-off between the security and the robustness against the noise (see \Cref{thm:ow-vra-test}), which is our main result. The outline of our proof is as follows. First, we combine the shortened codes and the certified deletion property (\Cref{thm:CDP-BI20}) to obtain \Cref{thm:asym-threshold}. Then, we use \Cref{thm:asym-threshold} to prove our main result (\Cref{thm:ow-vra-test}) by reducing the security of the error-tolerant SKL to the certified deletion property.

In \Cref{sec:ECC}, we link the threshold $\delta$ and the ECC in our PKE-SKL protocol to the Certified Deletion Property. In \Cref{sec:def}, we state the definition of the error-tolerant PKE-SKL. In \Cref{sec:cons}, we state the construction formally and reduce the security to the certified deletion property. In \Cref{sec:pke-skl-ham}, we use our main result to analyze the robustness of a specific SKL protocol.

\subsection{Feasible ECC and Threshold}
\label{sec:ECC}
In this section, we adapt the certified deletion property (\Cref{thm:CDP-BI20}) by using a randomly chosen shortened code instead of a fixed code. The adapted certified deletion property is essential in the security proof.

\begin{theorem}
    \label{thm:asym-threshold}
    Let $\lambda \in \mathbb N$ be the security parameter. Let $\epsilon \in (0,1)$ and $\delta\in (0,1)$ be two real numbers. Let $\{C_\lambda\}_{\lambda \in \mathbb N}$ be an arbitrary codes family, where $C_\lambda$ is a $(L_{out}, L_{in}, D)$ code. Let $\Theta$ be a set of bit strings.

    We adapt the certified deletion property (\Cref{thm:CDP-BI20}) by using the shortened code $C_{\lambda, \theta}$ for the computational basis. The challenger computes the syndrome $s$ with the randomly chosen shortened code instead of a fixed code.
    When $\epsilon$, $\delta$, $\Theta$, and $C_\lambda$ satisfy the following requirements:
    \begin{enumerate}
        \item We assume the code family has a non-vanishing code rate. In other words, $L_{in} > \epsilon L_{out}$.
        \item $\Theta$ is the set of bit strings with Hamming weight $\epsilon L_{out}$.
        \item $L_{out}(1 - \epsilon) (1 - h(\delta)) > (L_{out} - L_{in})$,
    \end{enumerate}
    for any adversary, the winning probability is negligible.
\end{theorem}
The theorem above states the relation between the error-correction code and the threshold $\delta$. For any non-vanishing code, we can compute how many incorrect guesses about $x_{|\cal I}$ are permitted. In \Cref{sec:cons}, the theorem helps us to assess the decryption correctness and the verification correctness.

Below, we prove \Cref{thm:asym-threshold}.
\begin{proof}[Proof of \Cref{thm:asym-threshold}]
First, we discuss the validity of the three requirements in \Cref{thm:asym-threshold}. For any code family with a non-vanishing rate, it is possible to find a constant $\epsilon$ such that $\epsilon L_{out} < L_{in}$. Otherwise, the coding rate is vanishing. Whenever $\epsilon L_{out} < L_{in}$, we have $L_{out} - L_{in} < (1- \epsilon)L_{out}$ and there exists a $\delta$ such that $L_{out}(1 - \epsilon) (1 - h(\delta)) > (L_{out} - L_{in})$.

The outline of the proof is as follows. We note that the shortened codes have no more syndromes than the original code. This is concluded in \Cref{lem:syndrome-number}. Then, we leverage \Cref{thm:CDP-BI20} to prove the security of ${\sf CDBB84}_{C,\theta}(\delta, \lambda)$, where we upperbound the number of syndromes using that of the original code. This suffices to prove \Cref{thm:asym-threshold}.
\begin{lemma}
    \label{lem:syndrome-number}
    The number of syndromes of the shortened codes $\{C_\theta\}_{\theta \in \{0,1\}^*}$ is at most the number of syndromes of the original code $C$.
\end{lemma}
\begin{proof}
    Let $H$ be the parity check matrix of the original code. We denote the parity check matrix of the shortened code as $H^\prime$. The number of syndromes is $2^{{\rm rank}(H^\prime)}$. We show that the rank of the parity check matrix is non-increasing.

    For any $\theta \in \{0,1\}^*$, the matrix $H^\prime$ is obtained by removing the $i$th columns such that $\theta[i] = 1$. This operation does not introduce new vectors into the support of $H^\prime$. Thus, the rank of the new code is non-increasing.
\end{proof}

First, we remind the reader of the security of the Certified Deletion game (see \Cref{thm:CDP-BI20})
\begin{equation}
    \frac{1}{2}\sqrt{2^{-m(1-h(\delta + \nu))+\log |S| + 1}} + 2\epsilon(\nu).
\end{equation}
By \Cref{lem:syndrome-number}, we have that
\begin{equation}
    \log |S| \leq L_{out} - L_{in}
\end{equation}
where $2^{L_{out} - L_{in}}$ is the number of the original code's syndromes. Also, the number of bits encoded in the computational basis $m$ equals $L_{out}(1-\epsilon)$. By substituting $m$ and $\log |S|$, we obtain
\begin{equation}
    \begin{aligned}
        &\frac{1}{2}\sqrt{2^{-m(1-h(\delta + \nu))+\log |S| + 1}} + 2\epsilon(\nu)  \\
        \leq & \frac{1}{2}\sqrt{2^{-L_{out}(1-\epsilon)(1-h(\delta + \nu))+(L_{out} - L_{in}) + 1}} + 2\epsilon(\nu).
    \end{aligned}
\end{equation}
By the $3$rd requirement in \Cref{thm:asym-threshold}, we can take a sufficiently small constant $\nu$ and obtain that $-L_{out}(1-\epsilon)(1-h(\delta + \nu))+(L_{out} - L_{in}) + 1 = O(-L_{out})$. This implies that $\frac{1}{2}\sqrt{2^{-L_{out}(1-\epsilon)(1-h(\delta + \nu))+(L_{out} - L_{in}) + 1}}$ is negligible in $L_{out}$, and in $\lambda$ implicitly. As stated in \Cref{thm:CDP-BI20}, $\epsilon(\nu)$ is negligible in $L_{out}$ and $L_{in}$. Thus, the winning probability for any adversary is negligible and we complete the proof.
\end{proof}


\subsection{The definition of error-tolerant SKL}
\label{sec:def}
\begin{figure*}
    \centering
    \includegraphics[width=0.8\linewidth]{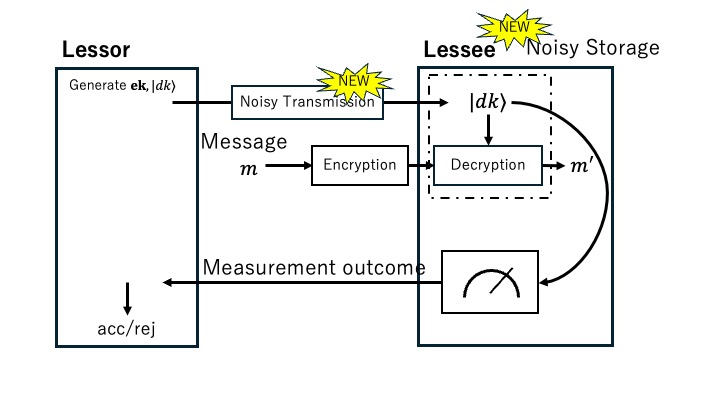}
    \caption{The figure above shows the outline of PKE-SKL. The lessor (the owner of the key) encodes its secret key into a quantum state $\ket{dk}$ and sends it to the lessee (who borrows the key). When the lessee holds the key, it can decrypt the ciphertext with the key for arbitrary times without destroying the key. In addition, when the lessor asks the lessee to return its key, the lessee measures the key and sends the outcome as the certificate of deletion. The lessor checks whether the certificate is valid.}
    \label{fig:PKE-SKL-DEF}
\end{figure*}
\begin{definition}[error-tolerant PKE-SKL]
\label{dfn:PKE-SKL}

Let $\lambda$ be the security paramter. In \Cref{fig:PKE-SKL-DEF}, we show the brief flow of PKE-SKL. Our definition differs from that in the prior works in that the key is subject to noise during transmission and storage. We state the correctness below:
\paragraph*{\bf Decryption Correctness}: When the key is affected by the noise, the probability that the original message $m$ and the decryption result $m^\prime$ differ is negligible.
\paragraph*{\bf Deletion Verification Correctness}: When the key is affected by the noise and the lessee deletes the key honestly, the probability that the lessor rejects is negligible.
\end{definition}

We state the security of PKE-SKL below.
\paragraph*{\bf IND-VRA Security:} When the lessee passes the verification,
\begin{enumerate}
    \item The lessor sends the verification key to the lessee.
    \item The lessee chooses two messages $m_0$ and $m_1$. The lessee sends the messages to the lessor. 
    \item The lessor samples $b \in \{0,1\}$ and encrypts $m_b$. The lessor sends the ciphertext to the lessee, and the lessee guesses the value of $b$.
\end{enumerate}
The PKE-SKL protocol is IND-VRA secure if the probability that the lessee makes a correct guess is at most $\frac{1}{2} + {\rm negl}(\lambda)$. We define the guessing probability to be $1/2$ when the lessee does not pass the verification.

Then, we state a weaker security definition below.
\paragraph*{\bf OW-VRA Security:} When the lessee passes the verification,
\begin{enumerate}
    \item The lessor sends the verification key to the lessee.
    \item The lessor chooses a message uniformly at random and encrypts it. Then, it sends the ciphertext to the lessee.
    \item The lessee guesses the value of $m$.
\end{enumerate}
The PKE-SKL protocol is OW-VRA secure if the probability that the lessee makes a correct guess is at most ${\rm negl}(\lambda)$.

We utilize the following lemma so that we only need to construct a protocol with the weaker security definition.
\begin{lemma}[\cite{APV23,KMY24}]
    For any PKE-SKL protocol that satisfies the OW-VRA security, it can be transformed into one with the IND-VRA security.
\end{lemma}
\noindent Because the OW-VRA security of our protocol is the same as that in the prior works, the lemma above is still available.

\subsection{The error-tolerant construction}
\label{sec:cons}
\begin{figure*}
    \centering
    \includegraphics[width=0.8\linewidth]{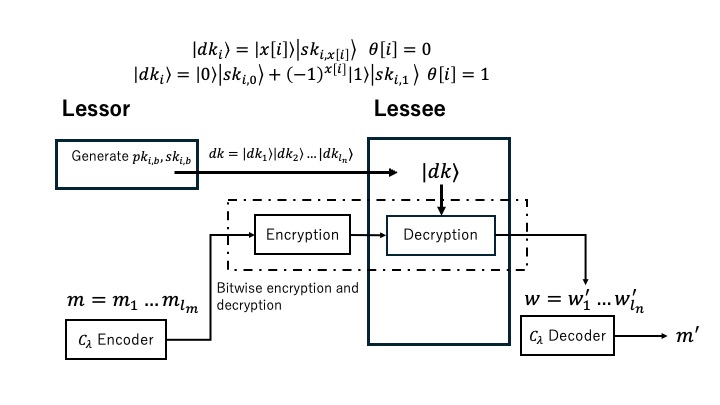}
    \caption{The leased key is $\ket{dk} = \ket{dk_1}\dots \ket{dk_{l_n}}$ where $\ket{dk_i}$ is encoded in the computational basis for $\theta[i] = 0$ and  Hadamard basis for $\theta[i] = 1$. When encrypting the message $m = m_1 \dots m_{l_m}$, the sender encodes $m$ into the codeword $w$ with the ECC $C_\lambda$. Then, the sender encrypts each bit of the codeword $w$ into $ct_{i,0}$ and $ct_{i,1}$ with the public keys $pk_{i,0}$ and $pk_{i,1}$, respectively. When decrypting, the lessee executes the classical decryption algorithm coherently to obtain each bit of the codeword. Then, the lessee decodes the (potentially corrupted) codeword to obtain the original message.}
    \label{fig:enc-dec}
    \includegraphics[width=0.8\linewidth]{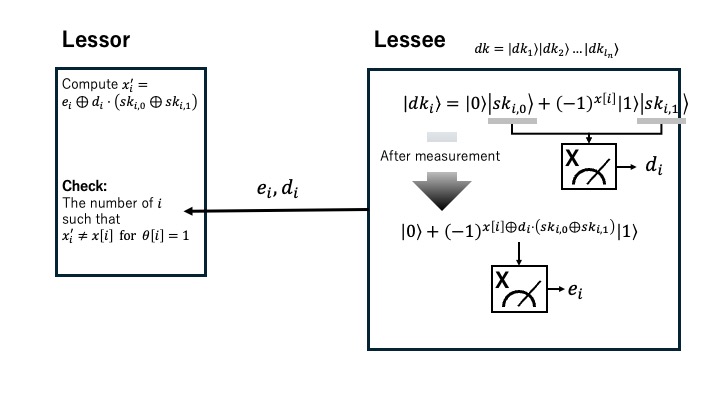}
    \caption{When the honest lessee deletes the key, it measures the last register in the Hadamard basis and obtains a bit string $d_i$. As a result, a Z error $Z^{d_i \cdot (sk_{i,0} \oplus sk_{i,1})}$ is applied to the remaining state. Then, the honest lessee measures the remaining qubit in the Hadamard basis and obtains a single bit $e_i$. The certificate is $(e_i,d_i)$. When no error occurs on the keys, the lessor cancels the Z error and obtains $x^\prime_i = x[i]$ for every $i$ such that $\ket{dk_i}$ is encoded in the Hadamard basis.}
    \label{fig:ourprotocol}
\end{figure*}
In this section, we present our error-tolerant PKE-SKL. Let $\lambda$ be the security parameter. The lessor and the lessee agree on an ECC $C_\lambda$ in advance. $l_n$ is the length of the codewords and $l_m$ is the length of the original message. The honest lessee proceeds as follows (see Fig.~\ref{fig:enc-dec}):
\begin{enumerate}
    \item The lessor samples $x, \theta \in\{0, 1\}^{l_n}$ where $\theta$ has Hamming weight exactly $\epsilon l_n$ uniformly at random. Then, it generates the public keys $pk_{i,b}$ and the secret key $sk_{i,b}$ for each $i\in[l_n]$ and $b\in\{0,1\}$. The leased key is $\ket{dk} = \ket{dk_1}\dots \ket{dk_{l_n}}$ where $\ket{dk_i} = \ket{x[i]}\ket{sk_{i,x[i]}}$ for $\theta[i] = 0$ and $\ket{dk_i} = [\ket{0}\ket{sk_{i, 0}} + (-1)^{x[i]}\ket{1}\ket{sk_{i, 1}}]/\sqrt{2}$ for $\theta[i] = 1$.
    \item The sender can choose a message $m$ and encrypt the message as follows. It encodes $m$ using a preshared ECC to obtain $w = C_\lambda(m)$. Then, it encrypts each bit of $w$ to obtain $\textsf{ct}_{i,b}$ for each $i$ and $b\in\{0,1\}$. The sender sends the ciphertext ($\textsf{ct}_{i,b}$)$_{i \in [l_n], b \in \{0, 1\}}$ to the lessee.
    \item The lessee decrypts each bit of the codeword $w$ using the following unitary $D_i$ on register ${\sf X}_i, {\sf PKE.DK}_i, {\sf OUT}_i$:
    $$ \ket{b}\ket{sk_{i,b}} \ket{v} \overset{D_i}{\mapsto}\ket{b} \ket{sk_{i,b}} \mid v \oplus \text{PKE.Dec}(sk_{i,b}, \textsf{ct}_{i,b}) \rangle $$ 
    where $b, v \in \{0, 1\}$ and ${sk_{i,b}} \in \{0, 1\}^{l_{dk}}$. The lessee measures the last register to obtain the outcome $w_i^\prime$. Let $w = w'_1 || \dots || w'_{l_n}$. The lessee decodes $w$ using the preshared ECC to obtain the original message.
    \item When the lessor asks the lessee to return its key, the lessee measures every bit of $\ket{dk_i}$ in the Hadamard basis to obtain an outcome $(e_i, d_i) \in \{0, 1\} \times \{0, 1\}^{l_{dk}}$ (see \Cref{fig:ourprotocol}). The lessee sends the certificate $\text{\sf cert} := (e_i, d_i)_{i \in [l_n]}$ to the lessor.
    \item The lessor rejects if the number of $i$s that satisfies the condition below is geater or equal to $\delta l_{H}$ where $l_{H} = \epsilon l_{n}$:
    $$e_i \neq x[i] \oplus d_i \cdot ({sk}_{i,0} \oplus {sk}_{i,1})$$ for $x[i]$ encoded in the Hadamard basis. Otherwise, the lessor accepts.
\end{enumerate}
For those who want to know the formal description of our construction, please refer to \Cref{sec:formal}.

\begin{theorem}
    \label{thm:ow-vra-test}
    For any $\epsilon, \Theta, \{C_\lambda\}_{\lambda \in \mathbb N}$ satisfying the three requirements in \Cref{thm:asym-threshold}, the PKE-SKL construction above satisfies OW-VRA security (see \Cref{dfn:pke-ow-vra}).
\end{theorem}
\begin{proof}
\begin{figure*}[htbp]
    \centering
    \includegraphics[width=0.8\linewidth]{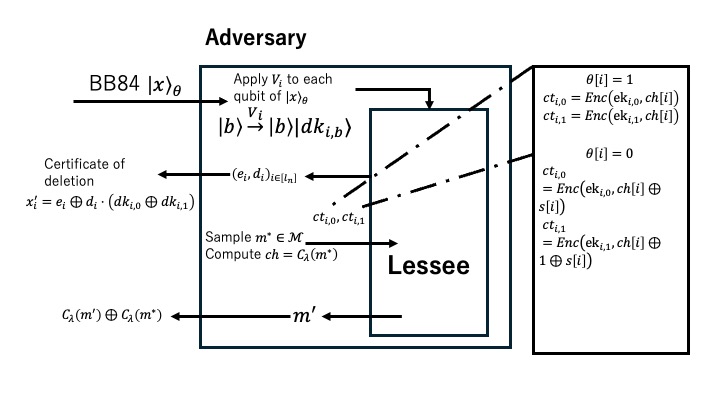}
    \caption{In the figure, we show the adversary that breaks the Certified Deletion Property using the cheating lessee. When the adversary receives a BB84 state $\ket{x}_\theta$, it samples classical PKE key pairs $(\text{\sf ek}_{i,b}, \text{\sf dk}_{i,b})$ for all $i \in [l_n]$ and $b \in \{0, 1\}$. Then, the adversary applies the isometry $V_i$ to each qubit of the BB84 state and sends the resulting state to the lessee. The resulting state is of the same form as the leased secret key. Then, the adversary transforms the certificate $(e_i, d_i)$ into $x_i^\prime$, where each $x_i^\prime$ is its gueess about $x[i]$. Once the adversary passes the first test in the game, the adversary receives $\theta, x_{|{\cal I}}, s \in \{0,1\}^{l_n}$ from the challenger. It samples the string $m^*$ and computes the query $ct_{i,0}, ct_{i,1}$ as shown in the figure above. Finally, the adversary receives $m^\prime$ and computes $C_\lambda(m^\prime) \oplus C_\lambda(m^*)$.}
    \label{fig:adversary}
\end{figure*}
In \Cref{fig:adversary}, we show that there exists an adversary that breaks the Certified Deletion Property if the lessee breaks the OW-VRA security.

To understand how the adversary works, we present a paradoxical adversary that generates the query using $x_{|\bar{\cal I}}$. It is paradoxical in the sense that it uses the information $x_{|\bar{\cal I}}$ to obtain the real value of $x_{\cal}$. $x_{|\bar{\cal I}}$ is the part of $x$ encoded in the computational basis. Then, we show that the non-paradoxical adversary in \Cref{fig:adversary} simulates the paradoxical one, which completes our proof.

The paradoxical adversary works like the adversary in \Cref{fig:adversary}, except it encrypts $ch[i]\oplus x[i] \oplus s[i]$ for $\theta[i] = 0$ where $ch$ is the challenge string. $x_{|\bar{\cal I}} \oplus s$ is a codeword of the shortened code. Here, we abuse the notation of the syndrome $s$ to represent its error vector as well. Then, the new challenge string $ch[i]\oplus x[i] \oplus s[i]$ is still a codeword of the original code $C_\lambda$. Below, we assume the cheating lessee to pass the verification and decrypt correctly with probability $1$. This simplifies our proof and does not compromise the generality. When the cheating lessee succeeds with non-negligible probability, the resulting adversary breaks the certified deletion property with non-negligible probability as well.

When the lessee produces $(e_i, d_i)_{i \in [l_n]}$ that passes the lessor's verification, the adversary transforms it into a certificate of deletion for BB84 state $\ket{x}_\theta$. This can be shown by expecting the lessor's verification process in \Cref{fig:ourprotocol}. The lessor transforms each $(e_i, d_i)$ into $x_i^\prime$ first. Then, it checks whether $x^\prime = x_1^\prime \cdots x_{l_n}^\prime$ is a valid certificate of deletion for $\ket{x}_\theta$. The adversary simulates the transformation process; thus, it generates a valid certificate that passes the first verification in the Certified Deletion game in \Cref{thm:asym-threshold}.

Then, we show that $C(m^\prime) \oplus C_\lambda(m^*)$ is a good guess for $x_{|\bar{\cal I}}$. When the lessee receives encrypted $ch[i]\oplus x[i] \oplus s[i]$, its view is the same as decrypting a normal ciphertext. Thus, the lessee should output a $m^\prime$ such that $C_\lambda(m^\prime)=ch\oplus x_{|\bar{\cal I}} \oplus s$. The adversary computes $C(m^\prime) \oplus C_\lambda(m^*)$ which equals to $x_{|\bar{\cal I}} \oplus s$. Here, we abuse the notation of the syndrome $s$ to represent its error vector as well. Thus, we obtain the following lemma.

\begin{lemma}
    The paradoxical adversary wins the Certified Deletion game with the same probability as the cheating lessee succeeds.
\end{lemma}

The paradoxical adversary uses the information $x_{|\bar{\cal I}}$, which is hidden. However, we show that the adversary in \Cref{fig:adversary} suffices. The difference between the adversary and the paradoxical one is the how $ct_{i,0}$ and $ct_{i,1}$ are generated for $\theta[i] = 0$. We encrypt $ch[i] \oplus (1-x[i]) \oplus s[i]$ to obtain $ct_{i,1-x[i]}$ instead. The reader may think the modified ciphertext still depends on $x[i]$. However, this is not the case. As shown in \Cref{fig:adversary}, the adversary obtains $ct_{i,0}$ by encrypting $ch[i] \oplus s[i]$ and $ct_{i,1}$ by encrypting $ch[i] \oplus 1\oplus s[i]$, regardless of the real value of $x[i]$.

For those $\ket{dk_{i}}$ such that $\theta[i] = 0$, the lessee only get the information about $dk_{i,x[i]}$. Thus, the lessee is unable to detect the change of the ciphertext $ct_{i,1-x[i]}$. Thus, we obtain the following lemma.
\begin{lemma}
    No QPT lessee can distinguish between the paradoxical adversary and the non-paradoxical adversary. For any unbounded lessee, it can simply recover the secret key using the public key.
\end{lemma}

Finally, we state that if the cheating lessee exists, the adversary is able to break the Certified Deletion Property (\Cref{thm:asym-threshold}) with almost the same probability, except for a negligible loss. We complete the proof.

\end{proof}

\subsection{Example of PKE-SKL using the concatenated Hamming codes}
\label{sec:pke-skl-ham}
In this section, we give a concrete example in which the protocol uses the concatenated Hamming code to encrypt the message. We discuss the upper bound of the threshold $\delta$ under this case. We define the concatenated Hamming codes as follows.
\begin{definition}
    Let $n \in \mathbb N$. The concatenated Hamming codes are $(7n, 4n, 3)$ codes. We describe the encoder and the decoder below.
    \paragraph*{Encoder} For any $m \in \{0,1\}^{4n}$, the encoder 
    \begin{itemize}
        \item Divides $m$ into blocks of $4$ bits such that $m = m_1 || m_2 || \dots || m_n$.
        \item Outputs $w = w_1 || w_2 || \dots || w_n$ such that $w_i$ is the codeword of $m_i$.
    \end{itemize}
    \paragraph*{Decoder} For $w \in \{0,1\}^{7n} $, the decoder 
    \begin{itemize}
        \item Divides $w$ into blocks of $7$ bits such that $w = w_1 || w_2 || \dots || w_n$.
        \item Decrypts each block $w_i$ as the Hamming code does.
    \end{itemize}
\end{definition}

Below, we fix the ECC to the concatenated Hamming codes and $\epsilon$ to $0.5$. Then, we find the upper bound of the threshold $\delta$.

\begin{lemma}[The upper bound of $\delta$]
    \label{lem:upperbound}
    For any $\delta \in [0,1/2]$ such that $h(\delta) < 1/7$, the PKE-SKL utilizing the concatenated Hamming code and $\epsilon = 0.5$ is secure.
\end{lemma}
\begin{proof}
    We can prove the lemma using \Cref{thm:asym-threshold}. Substitute $L_{in}$ with $4n$ and $L_{out}$ with $7n$, we can see that the first requirement in \Cref{thm:asym-threshold} is satisfied. 
    
    Then, we can substitute $L_{out}, L_{in}, \epsilon$ in the third requirement with the concrete values. The left side becomes
    \begin{equation}
        \begin{aligned}
            &L_{out}(1-\epsilon)(1-h(\delta)) \\
            =&7n \cdot 0.5 \cdot(1-h(\delta)) \\
            =&3.5n - 3.5n h(\delta)
        \end{aligned}
    \end{equation}
    The right side becomes
    \begin{equation}
        \begin{aligned}
            &L_{out} - L_{in} \\
            =&3n 
        \end{aligned}
    \end{equation}
    When $3.5n - 3.5n h(\delta) > 3n$ is satisfied, we can transform the inequality and obtain an upper bound $h(\delta) < 1/7$.
\end{proof}

We discuss the robustness against the errors below. Let the secret key be $\ket{dk_1}\ket{dk_2} \dots \ket{dk_{7n}}$. When at most one $\ket{dk_i}$ is corrupted, the protocol decrypts correctly. This is straightforward from the minimum Hamming distance of the code. In this case, at most one bit of the codeword is not decrypted correctly, and the concatenated Hamming code is able to correct the error. We want to point out that this is the worst case. If we divide the secret key into $n$ blocks, where each block corresponds to the block in the concatenated Hamming code, then the protocol decrypts correctly when there is at most one corrupted key in each block. Then, we discuss the verification correctness. We have shown the upper bound of $\delta$ in \Cref{lem:upperbound}, which is approximately $0.0203$. Thus, when the fraction of the corrupted secret key $\ket{dk_i}$ is below $0.0203$, the lessor certifies the deletion of the secret key correctly.

\section{Conclusion and Future works}
In this work, we constructed the first error-tolerant PKE-SKL scheme. We proved the security of our protocol and analyzed its robustness against errors. We utilized the tensored construction of the secret key in the prior works \cite{KMY24}. The construction makes it possible for the lessor to ensure the security while allowing the lessee to make wrong guesses for a fraction of bits in the Hadamard basis. Furthermore, we used the shortened code to connect the security of PKE-SKL and the error-tolerant certified deletion property~\cite{BI20}. Since shortened codes remain a relatively less explored concept in coding theory, our results not only benefit the field of PKE-SKL but also demonstrate a novel direction for their applications.

Below, we discuss some future directions.
\paragraph*{\bf A tighter analysis of the security} In this work, we bound the number of syndromes using that of the original code. We point out that a rigorous analysis of the number of syndromes of the shortened codes will give a tighter upper bound of the threshold $\delta$. However, this is more of interest in coding theory rather than quantum information.

\paragraph*{\bf Direct construction of IND-VRA secure PKE-SKL} In this work, we construct IND-VRA secure PKE-SKL by constructing OW-VRA secure PKE-SKL first and transforming the scheme later. This results in the drawback that the final protocol encrypts a one-bit message only, which makes the secret key grows linear to the length of the messages. We left open the problem of constructing an IND-VRA secure PKE-SKL with shorter keys. A potential solution to this problem is to construct an IND-VRA secure PKE-SKL directly, which reduces the security to the indistinguishable Certified Deletion Property in \cite{BI20}.

\paragraph*{\bf Error-tolerant SKL beyond PKE} Pseudorandom functions (PRFs) are another important cryptographic primitive. In \cite{KMY24}, they realized SKL for PRFs. One may think of the naive way as follows. The lessor and the lessee choose a common ECC $C$. When evaluating the output of the PRF, the lessor and lessee correct the raw output with $C$ and use the codeword as the final output. However, since the raw output of PRFs resembles a uniform random distribution, a small amount of error may change the raw output such that it is corrected to a completely different codeword. A solution to this problem may be using the Pseudorandom Codes (PRCs) instead of PRFs. SKL for digital signature uses the SKL for PRF as a component. Thus, it suffers from the same difficulty.

\bibliography{refs}

@PREAMBLE{
 "\providecommand{\noopsort}[1]{}" 
 # "\providecommand{\singleletter}[1]{#1}%" 
}

@book{IntroToModernCrypt,
  title={{Introduction to modern cryptography: principles and protocols}},
  author={Katz, Jonathan and Lindell, Yehuda},
  year={2007},
  publisher={Chapman and hall/CRC}
}

@misc{KN22,
      title={{Functional Encryption with Secure Key Leasing}}, 
      author={Fuyuki Kitagawa and Ryo Nishimaki},
      year={2022},
      eprint={2209.13081},
      archivePrefix={arXiv},
      primaryClass={quant-ph},
      url={https://arxiv.org/abs/2209.13081}, 
}

@misc{AKN+23,
      title={{Public Key Encryption with Secure Key Leasing}}, 
      author={Shweta Agrawal and Fuyuki Kitagawa and Ryo Nishimaki and Shota Yamada and Takashi Yamakawa},
      year={2023},
      eprint={2302.11663},
      archivePrefix={arXiv},
      primaryClass={quant-ph},
      url={https://arxiv.org/abs/2302.11663}, 
}

@misc{CGJL23,
      author = {Orestis Chardouvelis and Vipul Goyal and Aayush Jain and Jiahui Liu},
      title = {{Quantum Key Leasing for {PKE} and {FHE} with a Classical Lessor}},
      howpublished = {Cryptology {ePrint} Archive, Paper 2023/1640},
      year = {2023},
      url = {https://eprint.iacr.org/2023/1640}
}

@article{KNP25b,
  title={Collusion-resistant quantum secure key leasing beyond decryption},
  author={Kitagawa, Fuyuki and Nishimaki, Ryo and Pappu, Nikhil},
  journal={arXiv preprint arXiv:2510.04754},
  year={2025}
}

@misc{APV23,
      author = {Prabhanjan Ananth and Alexander Poremba and Vinod Vaikuntanathan},
      title = {{Revocable Cryptography from Learning with Errors}},
      howpublished = {Cryptology {ePrint} Archive, Paper 2023/325},
      year = {2023},
      url = {https://eprint.iacr.org/2023/325}
}

@misc{KMY24,
      title={{A Simple Framework for Secure Key Leasing}}, 
      author={Fuyuki Kitagawa and Tomoyuki Morimae and Takashi Yamakawa},
      year={2024},
      eprint={2410.03413},
      archivePrefix={arXiv},
      primaryClass={quant-ph},
      url={https://arxiv.org/abs/2410.03413}, 
}

@misc{PWY+25,
      author = {Duong Hieu Phan and Weiqiang Wen and Xingyu Yan and Jinwei Zheng},
      title = {{Adaptive Hardcore Bit and Quantum Key Leasing over Classical Channel from {LWE} with Polynomial Modulus}},
      howpublished = {Cryptology {ePrint} Archive, Paper 2025/099},
      year = {2025},
      doi = {https://doi.org/10.1007/978-981-96-0947-5_7},
      url = {https://eprint.iacr.org/2025/099}
}

@inproceedings{BJ24,
  title={Secret sharing with certified deletion},
  author={Bartusek, James and Raizes, Justin},
  booktitle={Annual International Cryptology Conference},
  pages={184--214},
  year={2024},
  organization={Springer}
}

@inproceedings{BVK+24,
  title={Software with certified deletion},
  author={Bartusek, James and Goyal, Vipul and Khurana, Dakshita and Malavolta, Giulio and Raizes, Justin and Roberts, Bhaskar},
  booktitle={Annual International Conference on the Theory and Applications of Cryptographic Techniques},
  pages={85--111},
  year={2024},
  organization={Springer}
}

@inproceedings{BKM+23,
  title={Weakening assumptions for publicly-verifiable deletion},
  author={Bartusek, James and Khurana, Dakshita and Malavolta, Giulio and Poremba, Alexander and Walter, Michael},
  booktitle={Theory of Cryptography Conference},
  pages={183--197},
  year={2023},
  organization={Springer}
}

@inproceedings{KNT23,
  title={Publicly verifiable deletion from minimal assumptions},
  author={Kitagawa, Fuyuki and Nishimaki, Ryo and Yamakawa, Takashi},
  booktitle={Theory of Cryptography Conference},
  pages={228--245},
  year={2023},
  organization={Springer}
}

@misc{BK22,
      author = {James Bartusek and Dakshita Khurana},
      title = {{Cryptography with Certified Deletion}},
      howpublished = {Cryptology {ePrint} Archive, Paper 2022/1178},
      year = {2022},
      url = {https://eprint.iacr.org/2022/1178}
}

@inbook{BI20,
   title={{Quantum Encryption with Certified Deletion}},
   ISBN={9783030643812},
   ISSN={1611-3349},
   DOI={10.1007/978-3-030-64381-2_4},
   booktitle={Theory of Cryptography},
   publisher={Springer International Publishing},
   author={Broadbent Anne and Islam Rabib},
   year={2020},
   pages={92–122} 
}

@misc{KNP25,
      author = {Fuyuki Kitagawa and Ryo Nishimaki and Nikhil Pappu},
      title = {{{PKE} and {ABE} with Collusion-Resistant Secure Key Leasing}},
      howpublished = {Cryptology {ePrint} Archive, Paper 2025/262},
      year = {2025},
      url = {https://eprint.iacr.org/2025/262}
}

@article{KRS09,
   title={The Operational Meaning of Min- and Max-Entropy},
   volume={55},
   ISSN={0018-9448},
   url={http://dx.doi.org/10.1109/TIT.2009.2025545},
   DOI={10.1109/tit.2009.2025545},
   number={9},
   journal={IEEE Transactions on Information Theory},
   publisher={Institute of Electrical and Electronics Engineers (IEEE)},
   author={Konig, Robert and Renner, Renato and Schaffner, Christian},
   year={2009},
   month=sep, pages={4337–4347} }

@article{Wie83,
  title={Conjugate coding},
  author={Wiesner, Stephen},
  journal={ACM Sigact News},
  volume={15},
  number={1},
  pages={78--88},
  year={1983},
  publisher={ACM New York, NY, USA}
}

@article{BL20,
  title={Uncloneable quantum encryption via oracles},
  author={Broadbent, Anne and Lord, S{\'e}bastien},
  journal={arXiv preprint arXiv:1903.00130},
  year={2019}
}

@inproceedings{Aar09,
  title={Quantum copy-protection and quantum money},
  author={Aaronson, Scott},
  booktitle={2009 24th Annual IEEE Conference on Computational Complexity},
  pages={229--242},
  year={2009},
  organization={IEEE}
}

@misc{TX25,
      title={Computational Certified Deletion Property of Magic Square Game and its Application to Classical Secure Key Leasing}, 
      author={Yuki Takeuchi and Duo Xu},
      year={2025},
      eprint={2510.04529},
      archivePrefix={arXiv},
      primaryClass={cs.CR},
      url={https://arxiv.org/abs/2510.04529}, 
}

@misc{KLYY25,
      author = {Fuyuki Kitagawa and Jiahui Liu and Shota Yamada and Takashi Yamakawa},
      title = {A Unified Approach to Quantum Key Leasing with a Classical Lessor},
      howpublished = {Cryptology {ePrint} Archive, Paper 2025/1871},
      year = {2025},
      url = {https://eprint.iacr.org/2025/1871}
}

@article{ZKL+23,
  title={Interactive cryptographic proofs of quantumness using mid-circuit measurements},
  author={Zhu, Daiwei and Kahanamoku-Meyer, Gregory D and Lewis, Laura and Noel, Crystal and Katz, Or and Harraz, Bahaa and Wang, Qingfeng and Risinger, Andrew and Feng, Lei and Biswas, Debopriyo and others},
  journal={Nature Physics},
  volume={19},
  number={11},
  pages={1725--1731},
  year={2023},
  publisher={Nature Publishing Group UK London}
}

\appendix

\section{Proof of \Cref{thm:CDP-BI20}}
\label{sec:modification-of-BI20}
In this section, we fix the flaw in the proof of \Cref{thm:CDP-BI20} by \cite{BI20}. 

Before we start our proof, we introduce (smooth) min-entropy and some useful lemmas.

\begin{definition}[Min-entropy]
    Let $\rho_{AB}$ be a bipartite state. The min-entropy is defined as follows:
    \begin{equation}
        H_{\min}(A|B)_\rho \coloneqq \sup \{s \in \mathbb R | \exists \sigma_{B}, \rho_{AB} \leq 2^{-s} 1_A \otimes \sigma_B  \}.
    \end{equation}
\end{definition}

\begin{lemma}
    \label{lem:min-entropy-diff}
    Let $\rho_{AB} = \sum_x P(x) \ket{x}\bra{x} \otimes \sigma_x$ be a bipartite state which $H_{\min}(A|B)_\rho = s$. Then, we consider the following state
    \begin{equation}
        \rho_{ARB} = \sum_x P(x) \ket{x}\bra{x} \otimes \ket{r_x} \bra{r_x} \otimes \sigma_x
    \end{equation}
    where $r \in R$. Then, we have
    \begin{equation}
        H_{\min}(A|RB)_\rho \geq s - \log |R|.
    \end{equation}
\end{lemma}
\begin{proof}
    We prove the lemma above using the operational characterization of $H_{\min}$. Let $p_{\sf guess}(A|B)_\rho$ be the maximum winning probability of the following game:
    \begin{itemize}
        \item Alice and Bob share the state $\rho_{AB}$ in advance.
        \item Alice measures the system $A$ in the computational basis and obtain the outcome $x$. Bob applies a POVM $\{E_{x^\prime}\}_{x^\prime}$ to the system $B$ and obtain the outcome $x^\prime$.
        \item Alice and Bob win iff $x = x^\prime$.
    \end{itemize}
    It is well known that \cite{KRS09}
    \begin{equation}
        p_{\sf guess}(A|B)_\rho = 2^{-H_{\min}(A|B)_\rho}.
    \end{equation}

    We assume $H_{\min}(A|RB) < s - \log |R|$. Then, we show a strategy for Alice and Bob to win with probability larger than $2^{-s}$, which implies that $H_{\min}(A|B) < s$. This contradicts the preassumption. We show the strategy below.
    \begin{itemize}
        \item Alice and Bob share the state $\rho_{AB}$ in advance.
        \item Alice measures the system $A$ in the computational basis and obtains the outcome $x$. 
        \item Bob prepares $\sum_{r \in R} \frac{1}{|R|} \ket{r}\bra{r}$ in register $R$. Bob applies a POVM $\{E_{x^\prime}\}_{x^\prime}$ to the system $RB$ and obtain the outcome $x^\prime$.
        \item Alice and Bob win iff $x = x^\prime$.
    \end{itemize}
    We note that $r$ matches $x$ with probability $1/|R|$. Conditioned on matching, $x = x^\prime$ with probability larger than $2^{-s+\log |R|}$. Thus, the winning probability for the strategy is larger than $2^{-s}$. By the operational characterization of min-entropy, we can see that $H_{\min}(A|B) < s$. This completes the proof as described aforementioned.
     
\end{proof}
Then, we define the purified distance and the smooth min-entropy as in \cite{BI20}.
\begin{definition}
    Let $\rho_A$ and $\sigma_A$ be two (subnormalized) states. We define the generalized fidelity as follows.
    \begin{equation}
        F(\rho_A, \sigma_A) \coloneqq (\Tr[\sqrt{\sqrt \rho_A \sigma_A \sqrt \rho_A } ] + \sqrt{1 - \Tr \rho_A} \sqrt{1 - \Tr \sigma_A} )^2.
    \end{equation}
    Then, the purified distance is defined as 
    \begin{equation}
        P(\rho_A, \sigma_A) \coloneqq \sqrt{1 - F(\rho_A, \sigma_A)}.
    \end{equation}
\end{definition}
\begin{definition}
    Let $\rho_{AB}$ be a bipartite state. Let $\epsilon \in [0,\Tr \rho_{AB}]$. We define 
    \begin{equation}
        H_{\min}^\epsilon(A|B)_{\rho_{AB}} = \sup_{\tilde \rho_{AB}, P(\rho_{AB},\tilde \rho_{AB}) \leq \epsilon} H_{\min} (A|B)_{\tilde \rho_{AB}}.
    \end{equation}
\end{definition}

First, we remind the reader of the security game in \Cref{thm:CDP-BI20}.
\begin{table*}[htbp]
    \centering
    \begin{tabular}{c|p{8cm}}
       Parameter  &  How it is used\\
       $\Theta$ & The set of possible basis $\theta$\\
       $n$  & The number of qubits in the Hadamard basis \\
       $m$  & The number of qubits in the computational basis \\
       $C$ & The linear ECC of which the codewords are $m$ bits \\
       $\delta$ & A real number falls in $[0,1/2]$ which indicates the tolerance of the verification process
    \end{tabular}
    \caption{In this table, we explain the notations used in the analysis of the security game of the Certified Deletion game.}
    \label{tab:parameters-in-CD}
\end{table*}

\paragraph*{Prepare-and-Measure game} 
\begin{enumerate}
    \item The challenger samples $\theta \in \Theta$ and $x \in \{0,1\}^{n+m}$ uniformly at random. Then, the challenger computes the syndrome $s$ as follows:
    \begin{itemize}
        \item Let ${\cal I} \coloneqq \{i\in [l_n] | \theta[i] = 1\}$ be the indices in which the qubit is encoded in the Hadamard basis. Note that $l_n:=n+m$.
        \item The challenger computes $s = {\sf Synd}(x_{|\cal \bar I})$, which is the syndrome of $x_{|\cal \bar I}$ with ECC $C$. $x_{|\cal \bar I}$ is the part of $x$ which will be encoded in the computational basis.
    \end{itemize}
    The challenger sends $\ket{x}_{\theta}$ to $A_0$.
    \item $A_0$ outputs a classical string $x^* \in \{0,1\}^{n+m}$ and a state $st$.
    \item If $hw(x^*_{|\cal I} \oplus x_{|\cal I}) \geq n \delta$, the challenger aborts the game, and the adversary loses.
    \item The challenger sends $\theta, x_{|\cal I}, s$,  and $st$ to $A_1$. $A_1$ outputs a classical string $x^\prime \in \{0,1\}^{n+m}$.
    \item The adversary wins if $x^\prime_{|\cal \bar I} = x_{|\cal \bar I}$. Otherwise, it loses.
\end{enumerate}

Then, we define a purified version of the game.
\paragraph*{Entangled Game}
\begin{enumerate}
    \item The challenger prepares $n + m$ EPR pairs in the registers $A_1 B_1 \dots A_{n+m} B_{n+m}$. Then, the challenger sends the registers $B_1 \dots B_{n+m}$ to the adversary.
    \item $A_0$ outputs a classical string $x^* \in \{0,1\}^{n+m}$ and a state $st$.
    \item The challenger samples $\theta \in \Theta$ uniformly at random and measures the register $A_1 \dots A_{n+m}$ according to $\theta$. Let the measurement outcome be $x \in \{0,1\}^{n+m}$. The challenger computes $err$ as follows:
    \begin{itemize}
        \item Let ${\cal I} \coloneqq \{i\in [l_n] | \theta[i] = 1\}$ be the indices in which the qubit is encoded in the Hadamard basis. 
        \item The challenger computes $s = {\sf Synd}(x_{|\cal \bar I})$, which is the syndrome of $x_{|\cal \bar I}$ with ECC $C$. $x_{|\cal \bar I}$ is the part of $x$ which will be encoded in the computational basis.
    \end{itemize}
    \item If $hw(x^*_{|\cal I} \oplus x_{|\cal I}) \geq n \delta$, the challenger aborts the game, and the adversary loses.
    \item The challenger sends $\theta, x_{|\cal I}, s$  and $st$ to $A_1$. $A_1$ outputs a classical string $x^\prime \in \{0,1\}^{n+m}$.
    \item The adversary wins if $x^\prime_{|\cal \bar I} = x_{|\cal \bar I}$. Otherwise, it loses.
\end{enumerate}

We prove the following lemma to state the relation between the security and the ECC.
\begin{lemma}
    \label{lem:cd-lemma}
    Let $A=(A_0, A_1)$ be any (unbounded) adversary. Let $S$ be the set of all possible syndromes. The probability that the adversary wins in the entangled game is at most $\frac{1}{2}\sqrt{2^{-m(1-h(\delta + \nu))+\log |S| + 1}} + 2\epsilon(\nu) $, where $\epsilon(\nu) = \exp(\frac{-m n^2 \nu^2}{(n+m)(n+1)})$. In our certified deletion property (see \Cref{thm:asym-threshold}), only one-way security is required instead of indistinguishability. Thus, the probability is $\frac{1}{2}\sqrt{2^{-m(1-h(\delta + \nu))+\log |S| + 1}} + 2\epsilon(\nu)$ instead of $\frac{1}{2}+\frac{1}{2}\sqrt{2^{-m(1-h(\delta + \nu))+\log |S| + 1}} + 2\epsilon(\nu) $. The probability is negligible for sufficiently large $n,m$. 
\end{lemma}
Since the entangled game is the purified version of the prepare-and-measure game, we prove the security of the prepare-and-measure game at the same time.

We present the proof below.
\begin{proof}
    We prove the lemma by bounding the min-entropy between $x_{|\cal \bar I}$
    and the adversary's internal register, as \cite{BI20} did. However, they did not take into account the loss introduced by sending the syndrome $s$ to the adversary in Step $5$. We fix the flaw and present the proof below.

    First, we state some additional notations. In Step 3, the challenger stores $x_{|\cal I}$ in register $V$ and $x_{|\cal \bar I}$ in register $X$. Seemingly, we divide $x^*$ into $x^*_{|\cal I}$ and $x^*_{|\cal \bar I}$. The adversary stores $x^*_{|\cal I}$ in register $Y$ and $x^*_{|\cal \bar I}$ in register $W$. In addition, we assume the information about $\theta$ is stored in register $S_\theta$.

    Let $\rho_{AYVWBS_\theta}$ be the overall state before the challenger obtains $x_{|\cal I}$, conditioned on the adversary passing the test in Step 4. The state $\rho_{XYVWBS_\theta}$ is obtained by measuring $x_{|\cal \bar I}$ and storing the outcome in register $V$. Then, we import the following lemma from \cite{BI20}
    \begin{lemma}[Proposition 5.6 from \cite{BI20}]
        For any $\nu \in (0, \frac{1}{2}-\delta]$ such that $\epsilon(\nu)^2 < \Pr[\text{The adversary passes the test in Step 4}]$, 
        \begin{equation}
            H_{\min}^\epsilon(X|VWBS_\theta)_\rho \geq m(1-h(\delta + \nu))
        \end{equation}
        where 
        \begin{eqnarray}
            \epsilon(\nu) &\coloneqq& \exp(\frac{-mn^2\nu^2}{(m+n)(n+1)})\\
            h(x) &\coloneqq& -x\log x - (1-x) \log (1-x).
        \end{eqnarray}
    \end{lemma}
    By \Cref{lem:min-entropy-diff}, we obtain the following lemma, which states the min-entropy after the adversary receives the information about the syndrome.
    \begin{lemma}
        \begin{equation}
            H_{\min}^\epsilon(X|VWBS_\theta {O})_\rho \geq m(1-h(\delta + \nu)) - \log |S|,
        \end{equation}
    where the information about the syndrome is stored in $O$. 
    \end{lemma}
    Then, we prove the security of the entangled game, with almost the same argument that \cite{BI20} proves Corollary 5.8. The difference is that we do not perform the privacy amplification. This is because our certified deletion property requires only one-way security instead of indistinguishability.
    
\end{proof}

\section{Formal Definitions and Proofs}
\label{sec:formal}
\begin{definition}[The definition of PKE with IND-CPA security~\cite{IntroToModernCrypt}]
    \label{dfn:IND-CPA-PKE}
    Let $\lambda$ be the security parameter. A public-key encryption (PKE) scheme is a tuple of algorithms $({\sf PKE.KG}, {\sf PKE.Enc}, {\sf PKE.Dec})$ as follows:
    \paragraph*{${\sf PKE.KG}(1^\lambda) \rightarrow ({\sf pk}, {\sf sk})$} is a PPT algorithm which takes as input the security paramter. It outputs a pair of secret key {\sf sk} and public key {\sf pk}.
    \paragraph*{${\sf PKE.Enc}({\sf pk}, b)\rightarrow {\sf ct}_b$} is a PPT algorithm. The algorithm takes as input a public key ${\sf pk}$ and a single bit $b \in \{0,1\}$. It outputs the ciphertext ${\sf ct}_b$.
    \paragraph*{${\sf PKE.Dec}({\sf sk}, {\sf ct}_b) \rightarrow b^\prime$} is a QPT algorithm. The algorithm takes as input a secret key ${\sf sk}$ and a ciphertext ${\sf ct}_b$. It outputs a single bit $b^\prime$ as the decryption result.

    The PKE scheme must satisfy the correctness as follows.
    \paragraph*{\bf correctness:} For $b \in \{0,1\}$, we have
    \begin{equation}
        \Pr[
        b \neq b^\prime:
        \begin{aligned}
            &({\sf pk}, {\sf sk}) \leftarrow {\sf PKE.KG}(1^\lambda) \\
            & {\sf ct}_b \leftarrow {\sf PKE.Enc}({\sf pk}, b)\\
            & b^\prime \leftarrow {\sf PKE.Dec}({\sf sk}, {\sf ct}_b)\\
        \end{aligned}
        ] = {\rm negl}(\lambda){.}
    \end{equation}
    Then, the PKE scheme must satisfy IND-CPA security as well.
    \paragraph*{\bf IND-CPA security:} For any QPT adversary $A^\lambda$
    \begin{equation}
        \begin{aligned}
            |&
        \Pr[A^\lambda({\sf pk}, {\sf ct}_0)=1:\begin{aligned}
            &({\sf pk}, {\sf sk}) \leftarrow {\sf PKE.KG}(1^\lambda) \\
            &ct_0 \leftarrow {\sf PKE.Enc}({\sf pk}, 0) \\
            &ct_1 \leftarrow {\sf PKE.Enc}({\sf pk}, 1) \\
        \end{aligned}] \\
        &- \Pr[A^\lambda({\sf pk}, {\sf ct}_1)=1:\begin{aligned}
            &({\sf pk}, {\sf sk}) \leftarrow {\sf PKE.KG}(1^\lambda) \\
            &ct_0 \leftarrow {\sf PKE.Enc}({\sf pk}, 0) \\
            &ct_1 \leftarrow {\sf PKE.Enc}({\sf pk}, 1) \\
        \end{aligned}]
        | = {\rm negl}(\lambda){.}
        \end{aligned}
    \end{equation}
\end{definition}

We use an $\text{IND-CPA}$ secure classical $\text{PKE}$ scheme for single-bit messages $\text{PKE} = (\text{PKE.\cal KG}, \text{PKE.\sf Enc}, \text{PKE.\sf Dec})$ as the building block. Then, we state the parameters used in the protocol. Let $\lambda$ be the security parameter, $l_m = \omega(\log \lambda)$ be the message length, and $\ell = \ell(\lambda)$ be the length of a decryption key generated by $\text{PKE.\cal KG}(1^\lambda)$. {\bf Let $C: \{0,1\}^{l_m} \to \{0,1\}^{l_n}$ be a linear ECC}. Let $\delta\in [0,1/2]$ and $\epsilon \in [0,1]$ be arbitrary constants. We construct $\text{PKESKL}$ = (${PKESKL.\cal KG}$, ${PKESKL.\sf Enc}$, ${PKESKL.Dec}$, ${PKESKL.Del}$, ${PKESKL.\sf DelVrfy}$) for $l_m$-bit messages.

\paragraph*{PKESKL.${\cal KG}$$(1^\lambda) \to (\text{\sf ek}, \text{dk}, \text{\sf dvk})$} 
\begin{enumerate} 
\item Sample $x, \theta \in\{0, 1\}^{l_n}$ where $\theta$ has Hamming weight exactly $\epsilon l_n$ uniformly at random. 
\item Generate $(\text{\sf PKE.ek}_{i,b}, \text{\sf PKE.dk}_{i,b}) \leftarrow \text{\sf PKE.KG}(1^\lambda)$ for all $i \in [l_n]$ and $b \in \{0, 1\}$.
\item Generate the quantum state $\ket{dk_i}$ for $i \in [l_n]$ based on $\theta[i]$ and $x[i]$ as follows:
\begin{equation}
    \begin{aligned}
        &\ket{x[i]}\ket{\text{\sf PKE.dk}_{i,x[i]}}  \quad \text{if } \theta[i] = 0 \\
        & \frac{1}{\sqrt{2}} \left( \ket{0} \ket{{\sf PKE.dk}_{i,0}}  + (-1)^{x[i]} \ket{1} \ket{{\sf PKE.dk}_{i,1}} \right)  \quad \text{if } \theta[i] = 1.
    \end{aligned}
\end{equation}
\item Output the keys:
    \begin{itemize}
    \item Encryption key (classical): $\text{\sf ek} := (\text{\sf PKE.ek}_{i,b})_{i \in [l_n], b \in \{0, 1\}}$.
    \item Decryption key (quantum): $dk := \ket{dk_1} \otimes \cdots \otimes \ket{dk_{l_n}}$. 
    \item Delete verification key (classical): ${\sf dvk} \coloneqq (\theta, \{x_i\}_{i \in [l_n]:\theta[i]=1},\{dk_{i,b}\}_{i\in[l_n]:\theta[i]=1,b\in\{0,1\}})$.
    \end{itemize} 
\end{enumerate}
    
\paragraph*{PKESKL.{\sf Enc}$({\sf ek}, m) \to \text{ct}$} 
\begin{enumerate} 
    \item Parse $\text{\sf ek} = (\text{\sf PKE.ek}_{i,b})_{i \in [l_n], b \in \{0, 1\}}$. 
    \item Compute $s = C(m) $.
    \item For $i \in [l_n]$ and $b \in \{0, 1\}$, generate $\text{PKE.ct}_{i,b} \leftarrow \text{PKE.\sf Enc}(\text{\sf PKE.ek}_{i,b}, s[i])$. 
    \item Output the ciphertext $\text{ct} $ $\coloneqq$ ($\textsf{PKE.ct}_{i,b}$)$_{i \in [l_n], b \in \{0, 1\}}$. 
\end{enumerate}

\paragraph*{PKESKL.Dec$(\text{dk}, \text{ct}) \to m'$} 
\begin{enumerate} 
    \item Parse $dk = \ket{dk_1} \otimes \cdots \otimes \ket{dk_{l_n}}$. 
    \item For each $i \in [l_n]$, we define the unitary $D_i$ on register ${\sf X}_i, {\sf PKE.DK}_i, {\sf OUT}_i$ as follows: 
    \begin{widetext}
        $$ |u\rangle | \text{PKE.dk} \rangle |v\rangle \quad \overset{D_i}{\mapsto} \quad |u\rangle | \text{PKE.dk} \rangle \mid v \oplus \text{PKE.Dec}(\text{PKE.dk}, \textsf{PKE.ct}_{i,u}) \rangle $$ 
    \end{widetext}
    where $u, v \in \{0, 1\}$ and $\text{PKE.dk} \in \{0, 1\}^\ell$. 
    \item Apply $D_i$ to $|\text{dk}_i\rangle \otimes |0\rangle$ and measure the rightmost register to obtain the outcome $s'_i$. 
    \item Let $s = s'_1 || \dots || s'_{l_n}$. We compute $m'$ by decoding $s$.
    \item Output $m'$.
\end{enumerate}

\paragraph*{PKESKL.Del$(\text{dk}) \to \text{\sf cert}$} 
\begin{enumerate} 
    \item Parse $dk = \ket{dk_1} \otimes \cdots \otimes \ket{dk_{l_n}}$. 
    \item For every $i \in [l_n]$, measure every bit of $\ket{dk_i}$ in the Hadamard basis to obtain an outcome $(e_i, d_i) \in \{0, 1\} \times \{0, 1\}^\ell$.
    \item Output $\text{\sf cert} := (e_i, d_i)_{i \in [l_n]}$. 
\end{enumerate}

\paragraph*{PKESKL.{\sf DelVrfy}$(\text{\sf dvk}, \text{\sf cert}) \to \top/\perp$}
\begin{enumerate} 
    \item Parse $\text{\sf dvk}$  = ($\{x[i]\}_{i \in [l_n]: \theta[i]=1}$, $\theta$, $\{\text{\sf PKE.dk}_{i,0}$, $\text{\sf PKE.dk}_{i,1}\}_{i \in [l_n]: \theta[i]=1}$) and $\text{\sf cert} = (e_i, d_i)_{i \in [l_n]}$.
    \item Output $\bot$ if the number of $i \in [l_n]$ that satisfies the condition below and $\theta[i]=1$ is geater or equal to $\delta l_{H}$ where $l_{H} = \epsilon l_{n}$:
    $$e_i \neq x[i] \oplus d_i \cdot (\text{\sf PKE.dk}_{i,0} \oplus \text{\sf PKE.dk}_{i,1}){.}$$
    \item Otherwise, output $\top$. 
\end{enumerate}
For those who are unfamiliar with the notations, we explain that $\bot$ means the algorithm rejects the certificate and $\top$ means the algorithm accepts the certificate.

\begin{definition}[IND-VRA Security \cite{KMY24}] We say that a $\text{PKE-SKL}$ scheme is $\text{IND-VRA}$ secure if it satisfies the following requirement, formalized by the experiment $\text{\sf Exp}^{\text{ind-vra}}_{\text{PKESKL}, \mathcal{A}}(1^\lambda, \text{coin})$ between a QPT adversary $\mathcal{A}$ and the challenger:
\begin{enumerate} 
\item The challenger runs $(\text{\sf ek}, dk, \text{\sf dvk}) \leftarrow \mathcal{KG}(1^\lambda)$ and sends $\text{\sf ek}$ and $dk$ to $\mathcal{A}$. 
\item $\mathcal{A}$ sends $\text{\sf cert}$. If $\text{\sf DelVrfy}(\text{\sf dvk}, \text{\sf cert}) = \bot$, the experiment outputs $0$ and $1$ with even probability. Otherwise, the challenger sends $\text{\sf dvk}$ to ${\cal A}$. 
\item ${\cal A}$ sends $(m_0, m_1) \in \mathcal{M}^2$ to the challenger. 
\item The challenger generates $\text{ct}^* \leftarrow \text{\sf Enc}(\text{\sf ek}, m_{\text{coin}})$ and sends it to ${\cal A}$.
\item $\mathcal{A}$ outputs a guess $\text{coin}'$ for $\text{coin}$. We use $\text{coin}^\prime$ as the output of the experiment.
 \end{enumerate}
For any $\text{QPT}$ adversary $\mathcal{A}$, the advantage is negligible ($\leq \text{\rm negl}(\lambda)$): 
$$ 
\begin{aligned}
    \text{Adv}^{\text{ind-vra}}_{\text{PKESKL}, \mathcal{A}}(\lambda) := |\Pr \left[ \text{Exp}^{\text{ind-vra}}_{\text{PKESKL}, \mathcal{A}}(1^\lambda, 0) = 1 \right] - \\
    \Pr \left[ \text{Exp}^{\text{ind-vra}}_{\text{PKESKL}, \mathcal{A}}(1^\lambda, 1) = 1 \right]| = \text{\rm negl}(\lambda).    
\end{aligned}
$$ 
\end{definition}

\begin{definition}[OW-VRA Security \cite{KMY24}] 
\label{dfn:pke-ow-vra}
We say that a $\text{PKE-SKL}$ scheme $\text{PKESKL}$ is $\text{OW-VRA}$ secure if it satisfies the following requirement, formalized by the experiment $\text{\sf Exp}^{\text{ow-vra}}_{\text{PKESKL}, \mathcal{A}}(1^\lambda)$ between a QPT adversary $\mathcal{A}$ and the challenger:
\begin{enumerate} 
\item The challenger runs $(\text{\sf ek}, dk, \text{\sf dvk}) \leftarrow \mathcal{KG}(1^\lambda)$ and sends $\text{\sf ek}$ and $dk$ to $\mathcal{A}$. 
\item $\mathcal{A}$ sends $\text{\sf cert}$ to the challenger. If $\text{\sf DelVrfy}(\text{\sf dvk}, \text{\sf cert}) = \perp$, the experiment outputs $0$. Otherwise, the challenger chooses $m^* \leftarrow \mathcal{M}$, generates $\text{ct}^* \leftarrow \text{\sf Enc}(\text{\sf ek}, m^*)$, and sends $\text{\sf dvk}$ and $\text{ct}^*$ to $\mathcal{A}$.
\item $\mathcal{A}$ outputs $m'$. The challenger outputs $1$ if $m' = m^*$ and otherwise outputs $0$ as the final output of the experiment. 
\end{enumerate}
For any $\text{QPT}$ adversary $\mathcal{A}$, the advantage is negligible: $$ \text{Adv}^{\text{ow-vra}}_{\text{PKESKL}, \mathcal{A}}(\lambda) := \Pr \left[ \text{Exp}^{\text{ow-vra}}_{\text{PKESKL}, \mathcal{A}}(1^\lambda) = 1 \right] \leq \text{negl}(\lambda). $$ \end{definition}

We present a rigorous proof of \Cref{thm:ow-vra-test}.
\begin{proof}
To prove \Cref{thm:ow-vra-test}, we define a sequence of hybrid distributions as follows.

We define ${\sf Hyb}_0$ as follows:
\begin{enumerate}
    \item The challenger samples $x, \theta \in\{0, 1\}^{l_n}$ where $\theta$ has Hamming weight exactly $\epsilon l_n$ uniformly at random{.}
    \item The challenger generates classical PKE key pairs $(\text{\sf PKE.ek}_{i,b}, \text{\sf PKE.dk}_{i,b}) \leftarrow \text{PKE.\cal KG}(1^\lambda)$ for all $i \in [l_n]$ and $b \in \{0, 1\}$ .
    \item The challenger generates the quantum state $|\text{dk}_i\rangle$ for $i \in [l_n]$ based on $\theta[i]$ and $x[i]$ as follows:
    \begin{equation}
    \begin{aligned}
        &\ket{x[i]}\ket{\text{\sf PKE.dk}_{i,x[i]}}  \quad \text{if } \theta[i] = 0 \\
        & \frac{1}{\sqrt{2}} \left( \ket{0} \ket{{\sf PKE.dk}_{i,0}}  + (-1)^{x[i]} \ket{1} \ket{{\sf PKE.dk}_{i,1}} \right)  \quad \text{if } \theta[i] = 1.
    \end{aligned}
    \end{equation}
    The challenger sends the quantum decryption key $\text{dk} := |\text{dk}_1\rangle \otimes \cdots \otimes |\text{dk}_{l_n}\rangle$ to the adversary.
    \item Let $\textsf{ek} \coloneqq ({\sf PKE.ek}_{i,0}, {\sf PKE.ek}_{i,1})_{i \in [l_n]}$. The challenger sends $\textsf{ek}$ to $A$.
    \item $A$ sends a classical deletion certificate $\textsf{cert} = (e_i, d_i)_{i\in[l_n]}$ to the challenger.
    \item If the number of $i \in [l_n]$ that satisfies $\theta[i]=1$ and the condition below is geater or equal to $\delta l_{H}$ where $l_{H} = \epsilon l_{n}$:
    $$e_i \neq x[i] \oplus d_i \cdot (\text{\sf PKE.dk}_{i,0} \oplus \text{\sf PKE.dk}_{i,1})$$
    then the experiment outputs $0$ and aborts.
    \item The challenger chooses a message $m^* \leftarrow \mathcal{M}$. The challenger computes the encoded message $ch = C(m^*)$ using the error correcting code $C$. The challenger generates $\text{\sf PKE.ct}_{i,b} \leftarrow \text{\sf PKE.Enc}(\text{\sf PKE.ek}_{i,b}, ch[i])$ for each $i \in [l_n]$ and $b \in \{0, 1\}$. Let $\textsf{ct}^* := (\textsf{PKE.ct}_{i,b})_{i\in[l_n],b\in\{0,1\}}$. Let $\textsf{dvk}$ $\coloneqq $ ($\{x[i]\}_{i\in[l_n]:\theta[i]=1}$, $ \theta$, $\{\textsf{PKE.dk}_{i,0},\textsf{PKE.dk}_{i,1}\}_{i\in[l_n]:\theta[i]=1}$). The challenger sends $\textsf{dvk}$ and $\textsf{ct}^*$ to $A$.

    \item The adversary sends a guessed message $m'$ to the challenger. The challenger outputs $1$ if $m' = {\sf Dec}(ch)$ and otherwise outputs $0$ as the final output of the experiment. The reader may notice that it is equivalent for the adversary to guess $ch$ and {$m^\ast$} since $ch$ is the encoded {$m^\ast$}.
\end{enumerate}
The hybrid above is the same as the original security game of the PKE-SKL scheme.

\paragraph*{${\sf Hyb}_1$:} We define ${\sf Hyb}_1$ the same as ${\sf Hyb}_0$ except for:
\begin{itemize}
    \item The challenger computes $x_{|\cal \bar I}$ such that $x_{|\cal \bar I}[i] = x[i]$ for $\theta[i] = 0$ and $x_{|\cal \bar I}[i] = 0$ for $\theta[i] = 1$.
    \item The challenger computes the syndrome $s$ using the shortend code $C_\theta$. Then, it computes $y$ with $ x_{|\cal \bar I} \oplus s$. We abuse the notion of the syndrome $s$ to represent its error vector.
\end{itemize}

\begin{lemma}
    \label{lem:pke-hyb1-hyb0-negl}
    We have $\Pr[{\sf Hyb}_0 = 1] = \Pr[{\sf Hyb}_1 = 1] $.
\end{lemma}
\begin{proof}
    In ${\sf Hyb}_1$, the challenger computes two additional strings $y$ ans $s$. It does not affect the condition in the hybrids. Thus, it is obvious that $\Pr[{\sf Hyb}_0 = 1] = \Pr[{\sf Hyb}_1 = 1]$.

\end{proof}

\paragraph*{${\sf Hyb}_2$:} We define ${\sf Hyb}_2$ the same as ${\sf Hyb}_1$ except for:
\begin{itemize}
    \item Let $ch^\prime$ be defined as 
    \begin{equation}
        \begin{aligned}
            &ch^\prime[i] = ch[i] &\theta[i]=1 \\
            &ch^\prime[i] = ch[i] \oplus y[i] &\theta[i]=0. \\
        \end{aligned}
    \end{equation}
    \item The challenger computes ${\sf PKE.ct}_{i,b} \leftarrow {\sf PKE.Enc}({\sf PKE.ek}_{i,b}, ch^\prime[i])$.
    \item The challenger checks $m^\prime = {\sf Dec}(ch^\prime)$ in Step 8 instead. In the lemma below, we show that $ch^\prime$ is always a valid codeword.
\end{itemize}

\begin{lemma}
    \label{lem:pke-hyb2-hyb1-eq}
    We have $\Pr[{\sf Hyb}_2 = 1]=\Pr[{\sf Hyb}_1 = 1]$.
\end{lemma}
\begin{proof}
    Let $m \in \{0,1\}^{l_m}$ be the message such that $y=C(m)$. Since $y$ is a codeword of the shortened code $C_\theta$, which is a subcode of the original code $C$. By the linearity of $C$,  we can see that $ch^\prime = C(m^* + m)$. Since $m^*$ is sampled uniformly at random, $m^* + m$ is also sampled from the uniform distribution. As a result, the output distribution does not change before and after.
\end{proof}

\paragraph*{${\sf Hyb}_3$:} We define ${\sf Hyb}_3$ the same as ${\sf Hyb}_2$ except for:
\begin{itemize}
    \item The challenger computes ${\sf PKE.ct}_{i,1-x[i]} \leftarrow {\sf PKE.Enc}({\sf PKE.ek}_{i,1-x[i]}, ch[i] \oplus (1-y[i]))$ for $i \in [l_n]\ s.t. \ \theta[i]=0$.
\end{itemize}

\begin{lemma}
    \label{lem:pke-hyb3-hyb2-negl}
    We have $ |\Pr[{\sf Hyb}_3=1] - \Pr[{\sf Hyb}_2=1] |= {\rm negl}(\lambda)$.
\end{lemma}

We can prove \Cref{lem:pke-hyb3-hyb2-negl} by the IND-CPA security of PKE as in \cite{KMY24}. The difference between ${\sf Hyb}_3$ and ${\sf Hyb}_2$ is the message encrypted into ${\sf PKE.ct}_{i,1 -x[i]}$. Since the lessee does not have the information about ${\sf PKE.dk}_{i,1 -x[i]}$ for $i$ such that $\theta[i] = 0$, it cannot distinguish the change in the ciphertexts. The readers may recall that the key state for such $i$ is $\ket{x[i]}\ket{{\sf PKE.dk}_{i, x[i]}}$ in which ${\sf PKE.dk}_{i,1 -x[i]}$ is absent. 

Below, we show an adversary for the certified deletion property that simulates ${\sf Hyb}_3$, thus proving that $\Pr[{\sf Hyb}_3 = 1] = {\rm negl}(\lambda)$. When the adversary receives a BB84 state $\ket{x}_\theta$ which is stored in register $({\sf A_1}, \dots, {\sf A}_{l_n})$:
\begin{enumerate}
\item The adversary generates classical PKE key pairs $(\text{\sf ek}_{i,b}, \text{\sf dk}_{i,b})$ for all $i \in [l_n]$ and $b \in \{0, 1\}$.
\item The adversary applies the following isometry ${\sf V}_i$ to ${\sf A}_i$ for every $i \in [l_n]$:
\begin{equation}
    V_i\ket{b}_{{\sf A}_i} = \ket{b}_{{\sf A}_i} \ket{{\sf dk}_{i,b}}_{{\sf DK}_i}.
\end{equation}
\item Let $\textsf{ek} \coloneqq ({\sf ek}_{i,0}, {\sf ek}_{i,1})_{i \in [l_n]}$. The adversary sends $\textsf{ek}$ and (${\sf A}_1$, ${\sf DK}_1$, $\dots$, ${\sf A}_{l_n}$, ${\sf DK}_{l_n}$) to the lessee.
\item The adversary receives a classical deletion certificate $\text{cert} = (e_i, d_i)_{i\in[l_n]}$ from the lessee. Then, the adversary computes $x^\prime[i] = e_i \oplus d_i \cdot ({\sf dk}_{i,0} \oplus {\sf dk}_{i,1})$.
\end{enumerate}
The adversary sends $x^\prime = x^\prime[1]||\dots||x^\prime[l_n]$ as its first answer.

When $x^\prime$ matches $x$ on every qubit in the Hadamard basis except for $\delta l_{H}$ qubits, the adversary receives $\theta, x_{|\cal I}, s \in \{0,1\}^{l_n}$ from the challenger. Then, the adversary queries the lessee to obtain the information encoded in the computational basis as follows:
\begin{enumerate}
    \item The adversary chooses a message $m^* \leftarrow \mathcal{M}$ and computes the encoded message $ch = C(m^*)$ using the error correcting code $C$. The adversary encrypts $ch[i]\oplus b \oplus s[i]$ to obtain $\text{\sf ct}_{i,b}$ for each $i \in [l_n]$ and $b \in \{0, 1\}$. 
    \item The adversary sends ($\{x[i]\}_{i\in[l_n]:\theta[i]=1}$, $\theta$, \{$\textsf{dk}_{i,0}$, $\textsf{dk}_{i,1}\}_{i\in[l_n]:\theta[i]=1}$) and $(\textsf{ct}_{i,b})_{i\in[l_n],b\in\{0,1\}}$ to the lessee.
    \item The adversary receives a message $m'$ from the lessee and computes $C(m^*) \oplus C(m^\prime)$. Then, combining with the syndrome $s$, the adversary is able to recover the information in the computational basis $x_{|\bar{\cal I}}$.
\end{enumerate}
    Below, we show that the adversary simulates the hybrid ${\sf Hyb}_3$ perfectly. 
    
    First, we show that the adversary produces a valid guess of the Hadamard basis. Whenever the lessee produces a valid certificate $(e_i, d_i)$, the adversary outputs a valid first-round answer (see \Cref{thm:CDP-BI20}) $x^\prime[i] = e_i \oplus d_i \cdot ({\sf PKE.dk}_{i,0} \oplus {\sf PKE.dk}_{i,1})$. The reader can check the fact by inspecting PKESKL.${\sf DelVrfy}$. The verification algorithm actually transforms the certificate into a guess of the information of the Hadamard basis, and then, it checks whether the guess is correct.
    
    Then, we show that the adversary simulates the challenge to the adversary in ${\sf Hyb}_3$ correctly. The challenge ciphertext is ${\sf PKE.ct}_{i, x[i]} \leftarrow {\sf PKE.Enc}({\sf PKE.ek}_{i, x[i]}, ch[i] \oplus y[i])$ and ${\sf PKE.ct}_{i, 1 - x[i]} \leftarrow {\sf PKE.Enc}({\sf PKE.ek}_{i, 1-x[i]}, ch[i] \oplus (1 - y[i]))$. We point out that $y = x \oplus s$, thus we can transform the ciphertext into the form as follows:
    \[
    \begin{aligned}
        & {\sf PKE.ct}_{i, x[i]} \leftarrow {\sf PKE.Enc}({\sf PKE.ek}_{i, x[i]}, ch[i] \oplus x[i] \oplus s[i] \\
        & {\sf PKE.ct}_{i, 1 - x[i]} \leftarrow {\sf PKE.Enc}({\sf PKE.ek}_{i, 1-x[i]}, ch[i] \oplus (1 - x[i]) \oplus s[i]).
    \end{aligned}
    \]
    If we replace $x[i]$ and $1 - x[i]$ with $b$, we can see that the challenge ciphertext is simulated by $\text{\sf PKE.ct}_{i,b} \leftarrow \text{\sf PKE.Enc}(\text{\sf PKE.ek}_{i,b}, ch[i]\oplus b \oplus s[i])$ which does not rely on the value of $x_{|\cal \bar I}$. Then, when the lessee produces a valid $m^\prime$, $y = C(m^*) \oplus C(m^\prime)$ satisfies $y \oplus s = x_{|\cal \bar  I}$
    for $i \in [l_n]$ such that $\theta[i] = 0$. Here, we also abuse the notation of the syndrome and its error vector.
    
    The adversary wins the security game with at least the same probability that the lessee
    wins. This can be seen from the fact that the adversary simulates ${\sf Hyb}_3$, such that it can guess $x_{|\cal \bar I}$ using the value of $m^\prime$ as shwon above.
    Thus, we conclude $\Pr[{\sf Hyb}_3 = 1] \leq \Pr[\text{A wins ${\sf CDBB84}_{C,\theta}(\delta, \lambda)$}] = {\rm negl}(\lambda)$, where the last equality comes from \Cref{thm:asym-threshold}.

    By combining \Cref{lem:pke-hyb1-hyb0-negl}, \Cref{lem:pke-hyb2-hyb1-eq}, and \Cref{lem:pke-hyb3-hyb2-negl}, we conclude the proof of \Cref{thm:ow-vra-test}.
\end{proof}

\end{document}